\documentclass[conference]{IEEEtran}

\usepackage[square, numbers, comma, sort&compress]{natbib}

\usepackage{epstopdf}
\usepackage{graphicx}
\usepackage{caption}
\usepackage{subcaption}
\usepackage{amsthm}
\graphicspath{{figures/}}
\usepackage{amsmath}
\usepackage{amsfonts}
\usepackage{tikz}
\usepackage{verbatim}
\usetikzlibrary{arrows,shapes}

\ifCLASSINFOpdf
\else
\fi
\newtheorem{lem}{lemma}

\IEEEoverridecommandlockouts

\begin{document}
%
\title{Caching Improvement Using Adaptive User Clustering  \thanks{The research of M. Assaad has been partially funded by Huawei, France and "Fondation Supelec"}}

\author{\IEEEauthorblockN{Salah Eddine Hajri and Mohamad Assaad }
\IEEEauthorblockA{Laboratoire des Signaux et Systemes (L2S, CNRS), CentraleSupelec,
Gif-sur-Yvette, France\\
Email: \{Salaheddine.hajri,\; Mohamad.Assaad\}@centralesupelec.fr}

}


%

 \IEEEoverridecommandlockouts
 \IEEEpubid{\makebox[\columnwidth]{978-1-5090-1749-2/16/\$31.00 \copyright 2016 IEEE
 	 \hfill} \hspace{\columnsep}\makebox[\columnwidth]{ }}


\maketitle

\begin{abstract}
	
	In this  article  we  explore one of the most promising technologies for 5G wireless networks using an underlay small cell network, namely proactive caching. Using the increase in storage technologies and through studying the users behavior, peak traffic can be reduced through proactive caching of the content  that is most probable to be requested. We propose a new method, in which, instead of caching the most popular content,  the users within the network are clustered according to their content popularity and the caching is done accordingly. We present also a method for estimating the number of  clusters within the network based on the \emph{Akaike information criterion}. We analytically derive a closed form expression of the hit probability and we propose an optimization problem in which  the  small base stations association with  clusters is optimized.
\end{abstract}


%
\IEEEpeerreviewmaketitle

\section{Introduction}
With the exponential growth in traffic size over the wireless networks, mobile operators need to adapt their networks in order to be able to cope with the exploding demand. New generation networks will include a number of new technologies in order to enhance the spectral efficiency, such as Massive MIMO, heterogeneous networks and D2D communication. Small cells technology was  proposed as a mean to offload an important amount of data from the macro BS. However, this technology is not without shortcomings, since the deployment of small base stations require a very high back-haul capacity which is not cost efficient. Information-centric networks are also emerging as a promising technology in order to adapt to the new spectral efficiency requirements. Predicting users’ behavior and proactively caching popular content in the edge of the network show an important potential  in terms of backhaul savings and user experience improvement. The idea of caching in mobile cellular networks is getting more and more attention\cite{informationcentric} \cite{MATHA}. In \cite{algo}, inner and outer bounds were derived for the joint routing and caching problem in small cell networks. In \cite{limits} an information theoretic approach for the problem of distributed caching was investigated. In \cite{code} coded caching was investigated for content delivery networks with hierarchical caching. In \cite{fund} the perfomrance of caching in Device-to-Device networks was studied.
In \cite{dd}, the conflict between collaboration distance and interference was derived between D2D users to maximize frequency reuse by taking advantage of distributed content caching. In \cite{unk}, the importance of the users' popularity profile estimation was investigated together with  a transfer learning approach which was proposed in order to enhance the preference of profile estimation.
 A clustering approach for content caching was proposed in \cite{cont} in order to reduce the service delay and shown to perform better than the unclustered and random caching approach. In order to properly adapt to users' behavior, we also propose a clustering based caching. While in  \cite{cont} a  spectral clustering algorithm is used, in this paper, prior to user grouping, we use  the Akaike Information Criterion  in order to efficiently estimate the number of user clusters in the sense of file popularity. This mitigate the need for complex eigenvalue-based algorithms, which can be computationally prohibitive. While in \cite{cont} each cluster is associated to a different SBs where the caching learning procedure is done locally, we propose a scheme in which the clustering is performed in a global and adaptive way since the cluster  number is estimated periodically. We then proceed by optimizing the SBs fractions associated with each cluster in order to  further optimize the system performances.

\section{ Contribution and outcomes}
The contributions of our work are summarized as follows:
\begin{enumerate}
\item An adaptive clustering approach of caching: In this paper we formulate the caching problem in a scenario where stochastically distributed small base stations  cache the content for users with different file preference. Given the different users profiles, we propose a clustering scheme in which users are grouped according to their preference and we develop a clustering scheme that uses this heterogeneousness  in order to effectively caches the files so that the hit probability is maximized.
\item User clusters estimation: Since the network doesn't have an a priori information about the number of user groups, we estimate this parameter  using \emph{Akaike Information Criterion} in order to  effectively estimate  the most popular files of each group. This will enable the adaptive nature of the approach where the system will estimate periodically the number of behavior clusters in order to adapt to new users.
\item Hit probability optimization: After user clustering, the small base station will cache a portion of the most popular files of user groups. We develop an optimization problem in which we define the optimal fractions of the small base stations associated with each cluster.
 \end{enumerate}
The paper is organized as follows: We describe the system model in Section III and the hit probability  in Section IV.  User cluster estimation will be investigated in Section V then the optimization of the hit probability  in Section VI. Finally, in Section  VII numerical results are presented. 
\section{ SYSTEM MODEL AND PRELIMINARIES}
We consider a  cellular network consisting of SBSs, which are modeled
according to a Poisson point process (PPP) $ \phi_s$ with density $\lambda_s$.
The users locations are  modeled also as a Poisson point process $ \phi$ with density $\lambda$. 
We suppose that every small base station has a limited storage  capacity $M$ files. We suppose a file catalog $C=\left\lbrace1...F \right\rbrace $ containing $F$ files each one with length $L$.
We suppose that the users within the network have a heterogeneous file popularity distribution,each user $u$ having popularity vector $P_u=\left[ p_{1u}...p_{Fu}\right] $. We suppose also that the users can be grouped according to there interest into $N_c$ clusters. Each group requesting more a certain group of files.  Due to the limited caching capability of the small base station, we suppose that the SBS will cache the most $M$ popular files of a given cluster. For each cluster $i, i=1..N_c$ a fraction $x_i$ of the available SBs will be caching the corresponding $M$ popular files. The SBs fractions  $\left[x_1...x_{N_c} \right] $ are optimized to improve the system performances.  We suppose that a mobile user will be associated with the nearest small base station that caches the requested content and is located within a radius $R$ from the user terminal. $R$ is equivalent to the maximum range over which the requested content can be served. If the requested file of a given user is available in the cache of a SBS within a distance $R$  a cache hit event occurs. The performance of our system depends on different parameters such as $\lambda_s$,$R$,the number of the clusters $N_c$, the cache size and the correlation of users behaviors.

	\begin{figure}[h!]
		\centering
		\includegraphics[width=7cm,height=2.7cm]{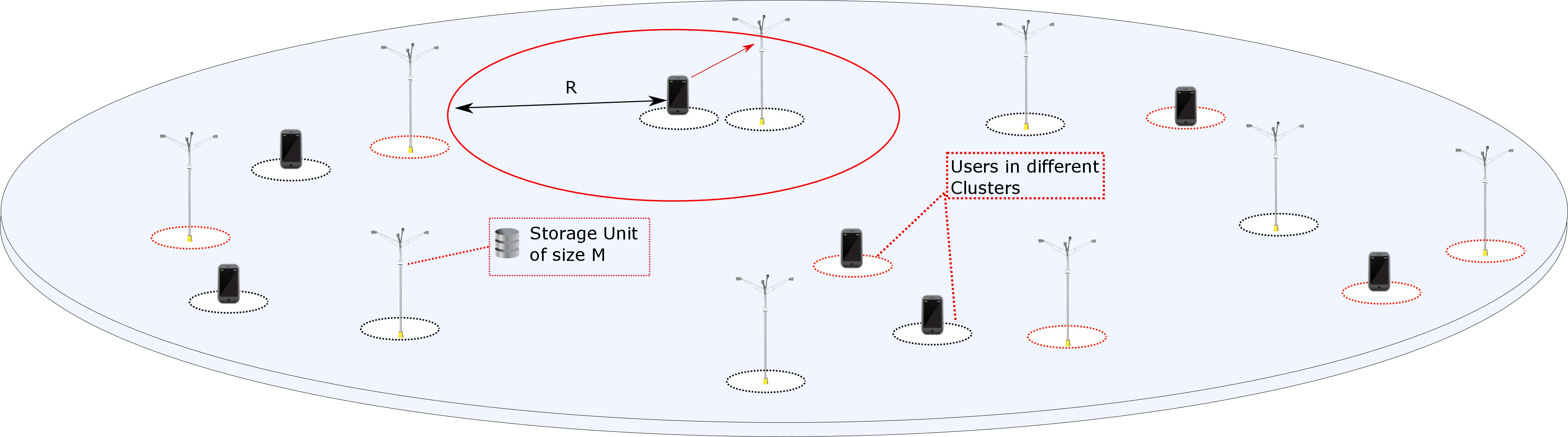}
		\label{System Model}
		\caption{ Semi-orthogonal User Clustering }
	\end{figure}
Increasing the small cell density would, obviously,  increase the probability of finding the requested file but it will cause  more energy consumption and increased interference.  
	  \section{Cache hit Probability }
	  In this paper,  we propose clustering the users according to their file demand. In order to assess the performance of our setting, we investigate the cache hit probability.
	  It is given by the probability of finding the requested file in the cache of a SBS within a radius $R$ from a given user \cite{modeling}. Our context differs from the one in \cite{modeling} since the users are clustered and each SBS caches a fraction of the popular files of a given cluster. Building on the result from \cite{modeling} and taking into consideration our model, the cache hit probability can be expressed as follows (the derivations are skipped for brevity):
	  \begin{align}
	  \mathbb P \left\lbrace cache hit\right\rbrace  = \frac{1}{N_u} \sum_{k=1}^{N_c} \sum_{u=1}^{N_u} \left( \sum_{i \in \Delta_k}^{} p_{iu} \right) \left( 1- e^{- x_k \lambda_s\pi R^2}\right) 
	  \end{align}
	  Where $N_c$ represent the number of user clusters, $\Delta_k$ represent the most $M$ popular files of cluster $k$. Since users belonging to the same cluster are not, necessarily, geographically co located, the SBs should be associated with the different clusters so that the most probable files to be requested are cached close to the users. The fractions $x_i, i=1.. N_c$ will ensure that the SBs are associated with the different clusters in a proportional way to the weight of the users belonging to these clusters. This equation ensures the probability of existence of at least one SBS storing the requested files within a radius $R$ from the mobile user. Each cluster is associated with a file popularity vector resulting from averaging the file popularities of the grouped users.  Once all the users are allocated to their respective clusters, the SBSs will be divided into $N_c$ groups. Each SBS will cache the most popular files of a given cluster. The number of SBS caching the $M$ most  popular files of cluster $k$ is given by $x_k \lambda_s\pi R^2$. The number of user clusters is unknown and should be estimated in order to have an efficient clustering and to adapt to any change in user behavior.  
\section{Cluster Estimation: Akaike information criterion}
\subsection{ Akaike information criterion}
While  users have normally different content popularity, some correlation in user requests may occur for people from the same social group for example.  In order to  efficiently cluster users according to content interest we need to estimate the number of clusters which is not an a priori information. Doing this  estimation periodically enables the system to  cope with terminal mobility or any change in user behavior. The aim is to group users so that the divergence in their request pattern is minimized. This means that users in the same cluster are most likely to  request the same content. In order to estimate the number of clusters while minimizing the variance of user profile in each group,  we use the \emph{AKAIKE INFORMATION CRITERION} (AIC)\cite{akaike}, which is a statistical model selection criterion. Given a collection of models for a set of data, AIC estimates the quality of each model, relative to each other.  AIC is based on \emph{KULLBACK-LEIBLER} information: it offers a relative estimate of the information lost when a given model is used to represent the process that generates the data.
  In doing so, it deals with the trade-off between the goodness of fit of the model based on maximum likelihood estimation and the complexity of the model given by the number of the estimable parameters. We consider a set of statistical models representing the process  generating the popularity vectors of the files for each user: $\varXi=\left\lbrace \xi_{{N_c}_{min}} ...\xi_{{N_c}_{max}}\right\rbrace $ where $\left\lbrace {N_c}_{min}...{N_c}_{max}\right\rbrace $ represent the range from which the number of clusters will be chosen and $\xi_{i}$ is the statistical model of the observed popularity vectors when assuming that the users are grouped into $i$ clusters. Each model is characterized by a set of finite parameters representing the variance of the popularity vectors and  the mean popularity vector in each cluster. The AIC of each model is given by: $ \mathrm{AIC}(\xi_{i}) = 2k_i - 2\ln(L_{\xi_{i}})$, $k_i$ being the number of estimable parameters in model $\xi_{i}$  and $L_{\xi_{i}}$ being the likelihood of the data according to the model $\xi_{i}$. In our case, each model is  characterized by 
 $i \times (F+1)$ estimable parameters,  $ i \times F $ centroid coordinates representing the average file popularity in each cluster  and $i$ variance estimate representing the divergence in user behavior in each cluster.
  The AIC can be understood thanks to  the \emph{KULLBACK-LEIBLER} information, also refereed to as discrepancy, which measures the information loss between two models. In fact, if we consider a generating model for the users popularity vectors characterized by the true number of user clusters $N_c$ and an approximation  model $\xi_{i}$ characterized by $N_i$,  the discrepancy between the two models is given by\cite{aic}:$$  d\left(N_i,N_c \right) =  \mathbb E\left\lbrace -2 \;log \; L(N_i |P )\right\rbrace  $$ 
where $L(N_i |P )$ represent the likelihood of the candidate model knowing the users popularity vectors. Evaluating $ d\left(N_i,N_c \right)$ is not possible since it requires the knowledge of the value of $N_c $. Akaike, however found that $-2 \;log \; L(N_i |P )$ is a biased estimator of the discrepancy measure and that after bias adjustment we have the following asymptotic result:
$$ \mathbb E\left\lbrace d\left(N_i,N_c \right)\right\rbrace \approx \mathbb E\left\lbrace 2k_i - 2 log(L_{\xi_{i}}) \right\rbrace $$ 
Note that this approximation becomes more accurate when increasing the number of the statistical process observations used in the likelihood function. Meaning that AIC offers an estimation of the expected cross entropy or discrepancy between the real model and the estimated one. In our case, we will compute the likelihood $L_{\xi_{i}}$   according to the Gaussian Mixture model, which is one of the most used assumptions when dealing with data clustering.  We model each cluster in $\xi_{i} $
 as a multi-variate Gaussian distribution, then  $log(L_{\xi_{i}})$ is given by \cite{xmean}:
	\begin{align}
& log\left( L_{\xi_{i}} \right)  = log\left(\prod_{u=1}^{N_u} \mathbb P(P_u) \right) \\
	\nonumber&= \sum_{u=1}^{N_u} \left( log(\frac{1}{\sqrt{2\pi} \hat{\sigma}_{\phi(u)}^F})-\frac{\lVert P_u - \hat{P}_{\phi(u)}  \rVert^2}{2\hat{\sigma}_{\phi(u)}^2} + log(\frac{N_{\phi(u)} }{N_u}) \right) 
	\end{align}
where $\phi(u)$ represent the cluster to which user $u$ is assigned.
The variance estimate in each cluster $k$ is given by: $$\hat{\sigma}_k^2= \frac{1}{(N_k)} \sum_{u \in k}^{} \lVert P_u - \hat{P}_{k}  \rVert^2 $$
 Then the log-likelihood of users popularity vectors is:
\vspace*{-2mm}
	\begin{align}
	log\left( L_{\xi_{i}} \right)  &=\sum_{k=1}^{i} 
	 -\frac{N_k}{2}( log(2 \pi) -1 + 2 log( \frac{N_k}{N_u})- F  log(\hat{\sigma}_k^2) )
	\end{align}
The AIC selects the model that minimizes discrepancy:
$$\xi_{AIC}=\underset{\xi}{\text{argmin}} \; AIC(\xi)$$ 	
 We use this criterion in user clustering according to popularity profile in order to derive the best user grouping while minimizing the loss in information regarding user file demand. It also allow for  a dynamic system that can adapt to any change in user behavior. 
\subsection{Proposed Clustering Algorithm}	
The proposed algorithm starts by clustering the users while assuming the existence of ${N_c}_{min}$ groups of users and then adds centroids according to some criterion until an upper bound ${N_c}_{max}$ is reached. If the AIC is decreasing over all the interval $\left[{N_c}_{min}...{N_c}_{max}  \right] $, the search should be extended until reaching  a minimum. In order to have an efficient clustering,  we need to derive a criterion to specify where the new centroids should be added. In fact clusters are represented by their centroid coordinates which specify the average popularity of each file for the users of the cluster.  It is obvious that in order to have good estimate of this popularity within the clusters, one would prefer to have minimum variance of file popularity. 

 We propose a simple criterion in which, at each step of the algorithm, a new centroid will be added in the cluster having the largest average distance between its centroid popularity vector and that of the users within. The new center will be selected as the node having the largest distance from its cluster centroid. This is due to the fact that large distance means that this user have a very different behavior which means that it has less probability of requesting the same files as the rest  of the users in its cluster. Once the new centroid is selected,  the algorithm will affect a user to the cluster for which the correlation between its  popularity vector and that of the centroid is maximized.  This will enable to minimize the disparity between users behaviors in the same cluster. At each iteration the popularity profile of the cluster is computed as the center-of-mass of all users popularity profiles:$\hat{P}_{k} = \frac{\sum_{u ,\phi(u)=k }^{} P_u  }{N_k}$. The proposed clustering algorithm is the following:
	  \begin{center}
	  	\begin{tabular}{ l  }
	  		User Clustering Algorithm\\
	  		\hline
	  			\hline
	  		\emph{Initialize}: Cluster number interval $\left[{N_c}_{min}...{N_c}_{max} \right]$\\
	  		                   choose randomly the first ${N_c}_{min}$ centroids from the users\\

	  		$1. $ Run ${N_{c_{min}}}-mean$ algorithm and compute $\mathrm{AIC}(\xi_{{N_c}_{min}})$\\
	  		$2. $ Choose the user having the greatest distance from its\\ cluster centroide in the cluster with the greatest variance add\\ a new centroide with popularity profile of the choosen user\\
	  		$3. $ Run steps $1$ and $2$ until reaching ${N_c}_{max} $\\
	  		$4. $ Choose the model which minimizes the AIC \\ 
	  		and cluster the users accordingly\\ 	
	  		\hline
	  		\hline
	  	\end{tabular}
	  \end{center}	 
\section{Cache Hit Probability Optimization}
 Once the users are clustered according to their traffic preference, we need to find the optimal value of the fractions $\left[x_1...x_{N_c} \right] $ that maximizes the cache hit probability. Each $x_i, i=1..N_c$ referring to the fraction of SBs associated with cluster $i$. To this end, we define the following optimization problem in which the cache hit probability is optimized over the fraction vector $X=\left[x_1...x_{N_c} \right]$:
\begin{equation*}
\begin{aligned}
&   \underset{X}{\text{maximize}}
& &  \frac{1}{N_u} \sum_{k=1}^{N_c} \sum_{u=1}^{N_u} \left( \sum_{i \in \Delta_k}^{} p_{iu} \right) \left( 1- e^{- x_k \lambda_s\pi R^2}\right) \\
& \text{subject to}
& & \sum_{k=1}^{N_c} x_k \leq 1
\end{aligned}
\end{equation*}
This optimization will allow to allocate more SBs to the most dominant Cluster while maintaining an adequate cached file diversity in the network. It turns out that this optimization problem is a concave problem which optimal solution can be derived using the KKT conditions.
\begin{lem}  		
The optimal fraction values are then given by:
\begin{align}
& x_s = \frac{N_c  log(\psi_s)-\sum_{k=1}^{N_c} log(\psi_k) +\lambda_s\pi R^2}{N_c\lambda_s\pi R^2} ,\; \forall  s=1..N_c\\ 
&\nonumber \text{where} \psi_s= \lambda_s\pi R^2\sum_{u=1}^{N_u} \left( \sum_{i \in \Delta_s}^{} p_{iu} \right) \\\nonumber
\end{align}
	
\end{lem}
\begin{proof}
The Lagrangian associated with this problem is given by:
\begin{align}
L\left( X,\mu\right)&= \frac{1}{N_u} \sum_{k=1}^{N_c} \sum_{u=1}^{N_u} \left( \sum_{i \in \Delta_k}^{} p_{iu} \right) \left( 1- e^{- x_k \lambda_s\pi R^2}\right)\\
\nonumber& - \mu\left( \sum_{k=1}^{N_c} x_k -1\right) 
\end{align}
where $\mu$ is the Lagrange multiplier. Taking the gradient of the Lagrangian we have the following:\\
\begin{align}
\nonumber& e^{- x_k \lambda_s\pi R^2} = \frac{N_u\mu}{\lambda_s\pi R^2\sum_{u=1}^{N_u} \left( \sum_{i \in \Delta_s}^{} p_{iu} \right)}, \;\forall s= 1..N_c  \Longrightarrow\\
\nonumber&  \lambda_s\pi R^2 = \sum_{k=1}^{N_c} log\left(  \lambda_s\pi R^2\sum_{u=1}^{N_u} \left( \sum_{i \in \Delta_k}^{} p_{iu} \right)\right) - log\left(N_u\mu \right)N_c\\\nonumber
\end{align}
The fraction are then given by:
\begin{align}
&  x_s = \frac{log\left(  \lambda_s\pi R^2\sum_{u=1}^{N_u} \left( \sum_{i \in \Delta_s}^{} p_{iu} \right)\right) -log\left(N_u\mu \right)}{\lambda_s\pi R^2}\\
\nonumber&  x_s = \frac{N_c  log(\psi_s)-\sum_{k=1}^{N_c} log(\psi_k) +\lambda_s\pi R^2}{N_c\lambda_s\pi R^2}
\end{align}
where: $\psi_s= \lambda_s\pi R^2\sum_{u=1}^{N_u} \left( \sum_{i \in \Delta_s}^{} p_{iu} \right) $
\end{proof}
We can clearly see from the expression obtained in Lemma 1 that the optimized system will tend to allocate more SBS to the cluster containing the most popular files among all users while maintaining a certain  diversity in the stored content.    
\section{Numerical Results AND Discussion}
We start by validating our analytical expressions and then investigate the impact of the different system parameters on the cache hit probability. We simulate the two PPP process of the users and the SBSs  over an area of $36\; Km^2$.
We generate the popularity profiles of the users randomly with a bias for a given set of files. Each cluster being characterized by a set of popular files. The users are allocated to the clusters randomly. In order to run the Clustering algorithm we only need a range over which the search of the number of clusters is made. In our simulations we take $\left[{N_c}_{min}...{N_c}_{max} \right] = \left[5...50 \right]$
\begin{figure}[h!]
	\centering
	\includegraphics[width=8cm,height=6cm]{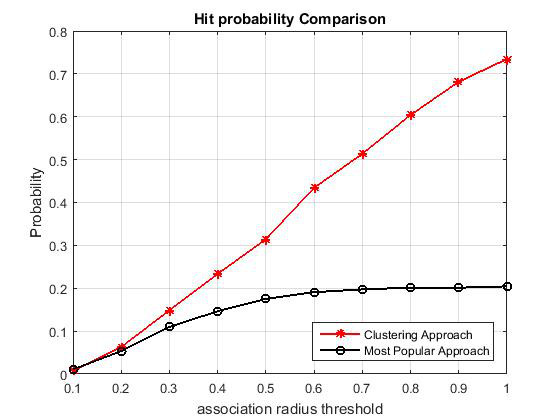}
	\caption{Hit Probability comparison}
\end{figure}
The Figure 2  shows the gain in term of cache hit probability when comparing the proposed clustering method and the scheme in which the most popular files in all the network are cached. We notice a considerable improvement due mainly to a better understanding of the user request pattern. At  a communication radius of $0.8\;Km$, the clustering approach triples the cache hit probability compared with the case where no clustering is used.
 \begin{figure}[h!]
 	\centering
 	\includegraphics[width=8cm,height=6cm]{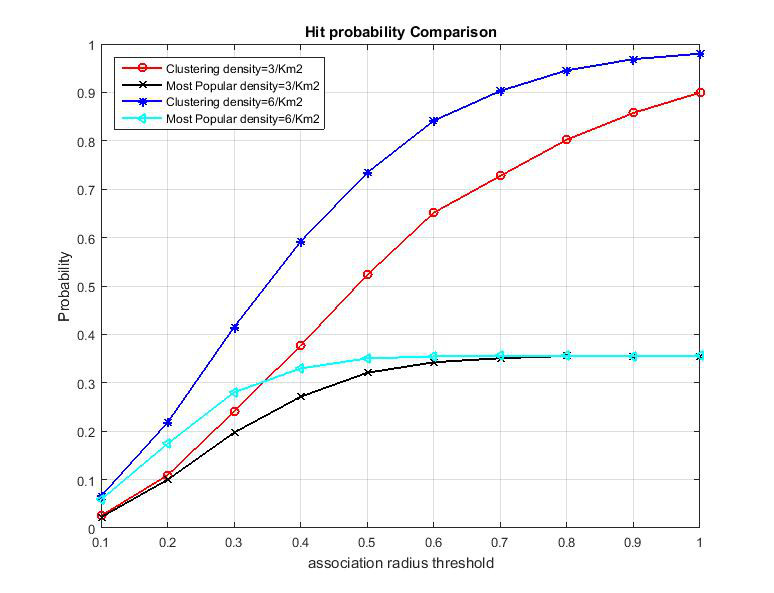}
 	\label{Average interference}
 	\caption{Hit Probability with different SBS densities }
 \end{figure}  
 Figure 3 shows the evolution of the cache hit probability with the SBS density. We notice that the hit probability for the scheme without clustering saturates at a certain value mainly because the same set of files is cached in all the SBSs which is clearly not an efficient approach. The hit probability for the clustering method increases mainly due to the fact that different files are cached in the SBS. Increasing the SBS  density brings these files closer to the users, hence improving the hit probability. In this setting, increasing the SBS density at a communication radius of $0.7\;Km$ improves the hit probability by $26 \%$ for the clustering approach while at the same radius, no improvement  was noticed for the classical approach. Figure 4 shows the impact of the cache size. Increasing the cache size is obviously a practical way to increase the system performances without the need to activate more costly SBSs. Figure 5 shows the performance of the AIC based model selection. It represents the average AIC normalized by the number of samples.  The lowest AIC value gives the model in which we minimize the loss  of information. The negative values of the AIC are due to a negative bias that characterize the AIC  with a small sample number.
\begin{figure}[h!]
	\centering
	\includegraphics[width=8cm,height=6cm]{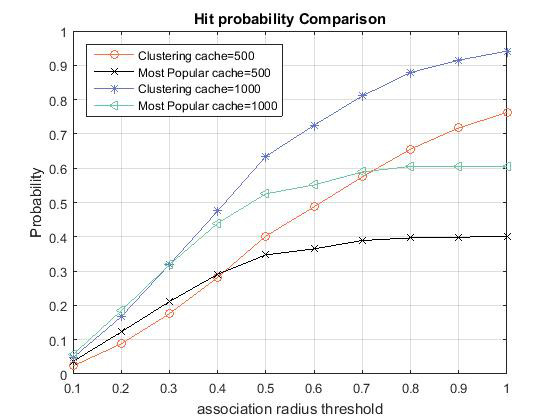}
	\label{Average interference}
	\caption{Hit Probability with different cache sizes }
\end{figure}
\begin{figure}[h!]
	\centering
	\includegraphics[width=8cm,height=6cm]{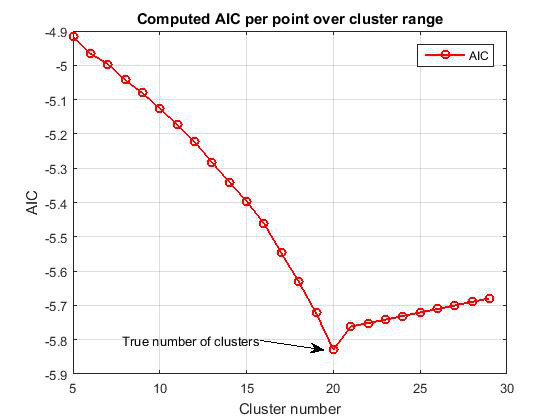}
	\label{Akaike information criterion}
	\caption{Akaike information criterion}
\end{figure}

\section{Conclusion AND Future Work}
We studied a cache enabled small cell network  with limited storage capability. We used a clustering approach in order to efficiently study the files popularity in the network and exploit the correlation between users demand. We develop an algorithm that enables an efficient clustering of users. We derived the hit probability when the SBS store the most popular content of the different user clusters and optimized SBS assignment. Finally, we performed numerical analysis which shows that the proposed algorithm outperforms the scheme in which no clustering is used.




\begin{thebibliography}{12}
				\bibitem {informationcentric}
				B. Ahlgren, C. Dannewitz, C. Imbrenda, D. Kutscher, B. Ohlman, \emph{A survey of information-centric networking}, IEEE Communications Magazine, Vol:50, Issue:7,  2012
				
				\bibitem {MATHA}
				M. Deghel, E. Bastug, M. Assaad, M. Debbah, \emph{On the benefits of Edge Caching for MIMO Interference Alignment}, 2015 IEEE 16th International Workshop on Signal Processing Advances in Wireless Communications 

				\bibitem {algo}
				K. Poularakis, G. Iosifidis, and L. Tassiulas, \emph{Approximation algorithms for mobile data caching in small cell networks}, IEEE Transactions on Communications  (Volume:62,  Issue: 10), August  2014
				
				\bibitem {limits}
				M. Maddah-Ali and U. Niesen, \emph{Fundamental limits of caching}, IEEE Transactions on Information Theory  (Volume:60, Issue: 5), Mars 2014
				
				\bibitem {fund}
				M. Ji, G. Caire, A. F. Molisch, \emph{Fundamental limits of caching in wireless D2D networks}, IEEE Transactions on Information Theory(Volume:62, Issue:2), Dec 2015

				\bibitem {code}
				N. Karamchandani, U. Niesen, M. A. Maddah-Ali, and S. N. Diggavi, \emph{Hierarchical coded caching}, 2014 IEEE International Symposium on Information Theory 	
				
				
			
				\bibitem {dd}
				N. Golrezaei, A. G. Dimakis, and A. F. Molisch, \emph{Wireless device-to-device communication with distributed caching}, 2012  IEEE International Symposium on Information Theory Proceedings (ISIT)
				
				\bibitem {unk}
				B. N. Bharath, K. G. Nagananda, \emph{Caching with Unknown Popularity Profiles in Small Cell Networks}, 2015 IEEE Global Communications Conference (GLOBECOM)	
				
				
				\bibitem {cont}
				M. S. ElBamby, M. Bennis, W. Saad, M. Latva-aho,	\emph{Content-Aware User Clustering and Caching in Wireless Small Cell Networks}, 2014 11th International Symposium on Wireless Communications Systems (ISWCS)
				
				\bibitem {modeling}
				S. Tamoor-ul-Hassan, M. Bennis, P. H. J. Nardelli, M. Latva-Aho, \emph{Modeling and Analysis of Content Caching in Wireless Small Cell Networks}, 2015 International Symposium on Wireless Communication Systems (ISWCS)		
					
				\bibitem {akaike}
				H. AIKAIKE,, \emph{A New Look at the Statistical Model Identification}, IEEE Transactions on Automatic Control (Volume:19, Issue:6), Dec 1974
				
				\bibitem {aic}
				J. E. Cavanaugh, \emph{Unifying the derivations for the Akaike and corrected Akaike information criteria}, Statistics \& Probability Letters, April 1997
					
				\bibitem {xmean}
				D. Pelleg, A. Moore, \emph{X-means: extending k-means with efficient estimation of the number of clusters}, ICML-2000

	
\end{thebibliography}
\end{document}